\documentclass[a4paper,russian]{article}
\usepackage[cp1251]{inputenc}
\usepackage[T2A]{fontenc}
\usepackage[english]{babel}
\usepackage[dvips]{graphicx}

\usepackage{amsmath,amssymb,amsthm}
\usepackage{hyperref}

\usepackage[top=20mm,right=20mm,left=20mm,bottom=20mm]{geometry}

\def\Ellips{{E}}
\def\B{\mathfrak{W}}
\def\C{\mathfrak{S}}
\def\ds{\displaystyle}
\def\bbR{\mathbb{R}}
\def\rang{\mathop{\rm rang}\nolimits}
\def\gs{\geqslant}
\def\ls{\leqslant}

\theoremstyle{plain}
\newtheorem{propos}{Proposition}

\begin{document}


\title{Integral manifolds of the reduced system\\ in the problem of inertial motion\\ of a rigid body about a fixed point\footnote{Submitted on January 9, 1974.}}

\author{M.P.\,Kharlamov\footnote{Moscow State University.}}

\date{}

\maketitle

\begin{center}
{\bf \textit{Mekh. Tverd. Tela} (Russian Journal ``Mechanics of Rigid Body''),\\
1976, No. 8, pp. 18--23}

\vspace{5mm}

\href{http://www.ics.org.ru/doc?pdf=1160\&dir=r}{http://www.ics.org.ru}

\vspace{2mm}

\href{https://www.researchgate.net/publication/258728370\_Integral\_manifolds\_of\_the\_reduced\_system\_in\_the\_problem\_of\_the\_inertial\_motion\_of\_a\_rigid\_body\_with\_a\_fixed\_point?ev=prf\_pub}{https://www.researchgate.net}

\end{center}

\vspace{3mm}

Let $A$, $B$ and $C$ $(A<B<C)$ be the principal moments of inertia of a rigid body having a fixed point. We consider the ellipsoid $\Ellips^2$ defined in $\bbR^3$ by the equation $Ax^2+By^2+Cz^2=1$ with the following metric on its surface
$$
d\Sigma=\sqrt{hABC}\left(A^2x^2+B^2y^2+C^2z^2\right)^{-\frac{1}{2}}d\sigma.
$$
Here $h>0$ is a constant, $d\sigma$ is the metric on $\Ellips^2$ induced by the scalar product of $\bbR^3$. The reduced system \cite{Kh1976} of the problem of the motion of a rigid body with a fixed point without external forces is equivalent to the Hamiltonian system on $T^*\Ellips^2$ with the Hamilton function
\begin{equation}\label{neq1}
  \ds{H=\frac{A^2x^2+B^2y^2+C^2z^2}{2ABC}(p_x^2+p_y^2+p_z^2)}.
\end{equation}
The projections to $\Ellips^2$ of the integral curves of this system with constant energy $H=h$ are, according to the Maupertuis principle, geodesics of the metric $d\Sigma$. Thus, instead of saying ``the basic integral curve of the Hamiltonian vector field $X_H$'' we use the term ``geodesic''.

\begin{propos} Any integral manifold $J_h=\{H=h\}$ in $T^*\Ellips^2$ is diffeomorphic to $T_1^*S^2$, i.e., to the bundle of the unit cotangent vectors to the sphere.
\end{propos}

\begin{proof} Let $D$ be the diffeomorphism of $\Ellips^2$ onto $S^2=\{\xi^2+\eta^2+\zeta^2=1\}$ such that $D(x,y,z)=(\sqrt{A}x,\sqrt{B}y,\sqrt{C}z)$. The corresponding diffeomorphism of the cotangent bundles $T^*D:T^*S^2\to T^*\Ellips^2$ has the form $x=\xi/\sqrt{A}$, $y=\eta/\sqrt{B}$, $z=\zeta/\sqrt{C}$, $p_x=\sqrt{A}p_\xi$, $p_y=\sqrt{B}p_\eta$, $p_z=\sqrt{C}p_\zeta$. In $\bbR^6$, the manifold ${M^3=}(T^*D)^{-1}(J_h)\subset T^*S^2$ is given by the equations
\begin{equation*}
\begin{array}{c}
  \ds{\frac{A^2x^2+B^2y^2+C^2z^2}{2h}(Ap_\xi^2+Bp_\eta^2+p_\zeta^2)} = 1, \quad
  \xi^2+\eta^2+\zeta^2=1, \quad
  \xi p_\xi + \eta p_\eta + \zeta p_\zeta = 0.
\end{array}
\end{equation*}
In turn, $T_1^*S^2$ is given by the equations
\begin{equation*}
\begin{array}{c}
  p_\xi^2+p_\eta^2+p_\zeta^2=1, \quad
  \xi^2+\eta^2+\zeta^2=1, \quad
  \xi p_\xi + \eta p_\eta + \zeta p_\zeta = 0.
\end{array}
\end{equation*}
Let us define $\vartheta:M^3\to T_1^*S^2$ by putting $\vartheta(\xi,\eta,\zeta,p_\xi,p_\eta,p_\zeta)=(\xi,\eta,\zeta,p_\xi/\delta,p_\eta/\delta,p_\zeta/\delta)$, where $\delta^2=p_\xi^2+p_\eta^2+p_\zeta^2$. The map $\vartheta$ is bijective and $\rang\vartheta = 3$ at each point of $M^3$. Therefore $\vartheta$ is a diffeomorphism. The composition $\alpha = \vartheta \circ T^*D^{-1}$ takes $J_h$ to $T_1^*S$. This proves the statement.
\end{proof}

Let us point out one more property of manifolds $J_h$. Let $Q$ be the closed ball in $\bbR^3$ of radius $\pi$ with the center at the coordinates origin. Declare the diametrically opposite points of the ball boundary equivalent and denote by $P$ the quotient space of the topological space $Q$ with respect to this equivalence. For each $\nu\in P$, we denote by $v_\nu \in SO(3)$ the element for which $\nu$ is the defining vector (see \cite{Kh1976}). Let $\omega_0\in T_1^*S^2$ have the coordinates $\xi=1$, $\eta=\zeta=0$, $p_\eta=1$, $p_\xi=p_\zeta=0$. The map $\beta:P\to J_h$ defined as $\beta(\nu)=(v_\nu \circ \alpha)^{-1}(\omega_0)$ is a homeomorphism. We use the map $\beta$ for a geometric interpretation.

Let $\lambda$, $\mu$ be the elliptic coordinates on $\Ellips^2$
\begin{equation*}
  \ds x^2=a\frac{(a-\lambda)(a-\mu)}{(a-b)(a-c)},\quad y^2=b\frac{(\lambda-b)(b-\mu)}{(a-b)(b-c)}, \quad z^2=c\frac{(\lambda-c)(\mu-c)}{(a-c)(b-c)} ,
\end{equation*}
where $a=1/A$, $b=1/B$, $c=1/C$. The elliptic coordinates change in the regions $a \gs \lambda \gs b \gs \mu \gs c$. Denote $F(t)=(a-t)(b-t)(c-t)/t$. The Hamilton function \eqref{neq1} takes the form
\begin{equation*}
  \ds{H=2\frac{\lambda\mu}{\lambda-\mu}|F(\lambda)p_\lambda^2-F(\mu)p_\mu^2|.}
\end{equation*}
Let us introduce on $\Ellips^2$ the Liouville coordinates by the formulas
\begin{equation*}
  \ds u=\int\limits_b^\lambda\frac{dt}{\sqrt{F(t)}},\qquad v=\int\limits_c^\mu\frac{dt}{\sqrt{-F(t)}}.
\end{equation*}
For them, the regions are
\begin{equation*}
  \ds 0\ls u\ls m = \int\limits_b^a\frac{dt}{\sqrt{F(t)}}, \qquad 0\ls v\ls n=\int\limits_c^b\frac{dt}{\sqrt{-F(t)}}.
\end{equation*}
In Fig.~\ref{fig_01}, we show parametric curves of $u$ and $v$ on the ellipsoid. In the coordinates $(u, v)$,
$$
H=2[V(v)-U(u))]^{-1}(p_u^2+p_v^2),
$$
where $U(u)=1/\lambda(u)$, $V(v)=1/\mu(v)$. Note that $dU/du=-\lambda^{-2}\sqrt{F(\lambda)}$, i.e., $dU/du=0$ at $u=0$, $u=m$ and $dU/du<0$ at $0<u<m$. Similarly, $dV/dv=0$ at $v=0$, $v=n$ and $dV/dv<0$ at $0<v<n$.

\begin{figure}[!ht]
\centering
\includegraphics[width=0.35\paperwidth,keepaspectratio]{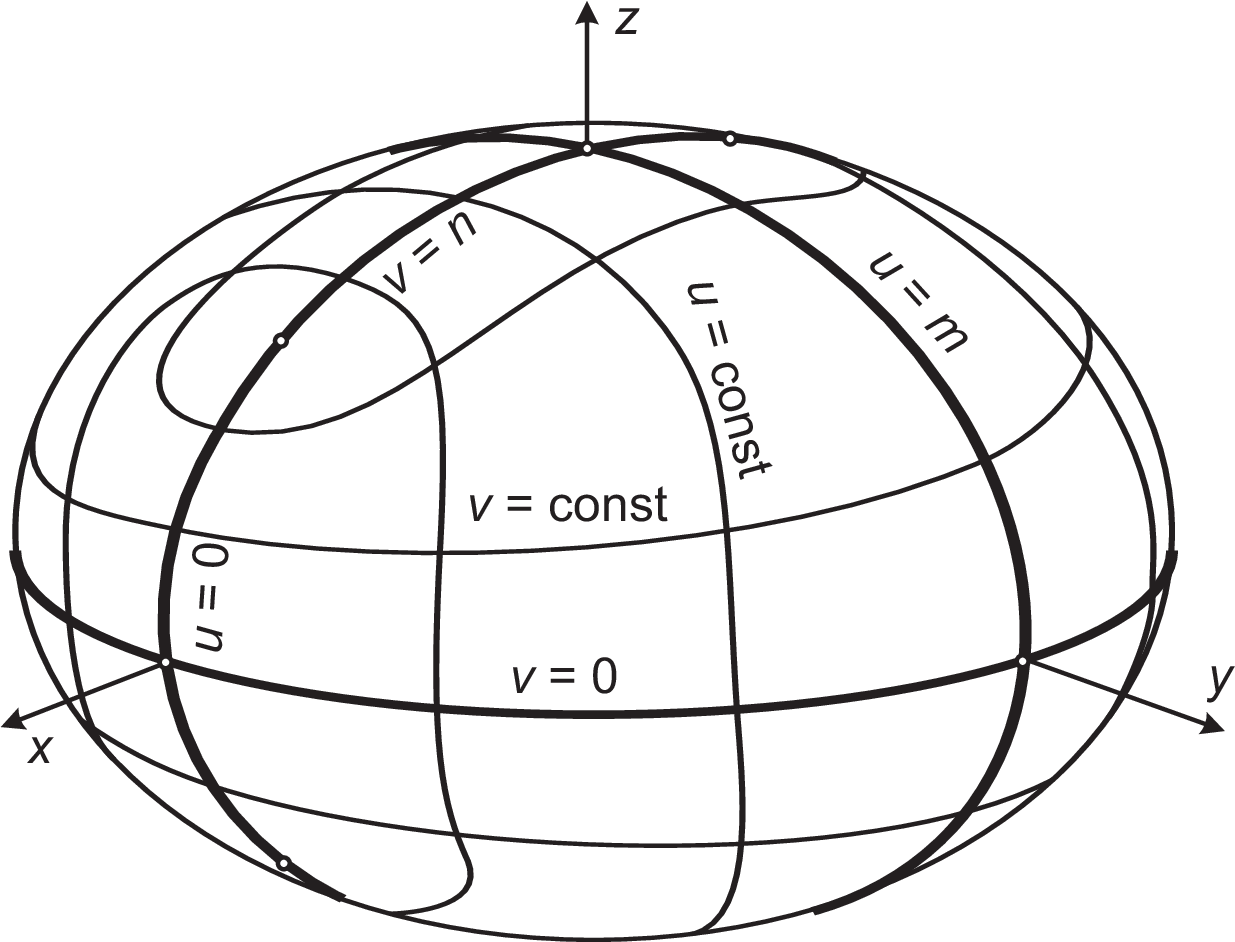}
\caption{Coordinates on the ellipsoid.}\label{fig_01}
\end{figure}

In the domain where $u$ and $v$ are local coordinates the restriction of the initial system to the manifold $J_h$ admits the integrals
\begin{equation}\label{neq2}
  p_u^2+hU(u)=h\kappa, \qquad p_v^2-hV(v)=-h\kappa.
\end{equation}
Denote by $J_{h, \kappa}$ the subset of $J_h$ defined by equations \eqref{neq2}. The admissible values of $\kappa$ are $A \ls \kappa \ls C$. Let us find out the topological type of the integral manifolds $J_{h, \kappa}$ in the following cases: 1)~$A \ls \kappa < B$; 2)~$B<\kappa \ls C$; 3)~$\kappa = B$.

Let $\B =\Ellips^2\setminus\{u=0\}$ and $\C =\Ellips^2\setminus\{v=n\}$ be the regions on the ellipsoid surface. In them, we introduce the local coordinates $\B =\{(w,\varphi\,{\rm mod}\,4n)\}$, $\C =\{(s,\theta\,{\rm mod}\, 4m)\}$ similar to cylindrical ones putting
\begin{equation*}
\begin{array}{cc}
  w=\begin{cases}
  u  &\text{при $x\ls 0$;}\\
  2m-u  &\text{при $x\gs 0$,}\\
  \end{cases}
  &
  s=\begin{cases}
  v  &\text{при $z\ls 0$;}\\
  -v  &\text{при $z\gs 0$,}\\
  \end{cases}
\end{array}
\end{equation*}

\begin{equation*}
\begin{array}{cc}
  \varphi=\begin{cases}
  v &\text{при $y\gs 0, z \gs 0$;}\\
  2n-v &\text{при $y\ls 0, z \gs 0$;}\\
  2n+v &\text{при $y\ls 0, z \ls 0$;}\\
  4n-v &\text{при $y\gs 0, z \ls 0$,}\\
  \end{cases}
  &
  \theta=\begin{cases}
  u &\text{при $x\gs 0, y \gs 0$;}\\
  2m-u &\text{при $x\ls 0, y \gs 0$;}\\
  2m+u &\text{при $x\ls 0, y \ls 0$;}\\
  4m-u &\text{при $x\gs 0, y \ls 0$.}\\
  \end{cases}
\end{array}
\end{equation*}
It is easily shown that these coordinates are compatible with the smooth structure of the ellipsoid.

Let us consider the cases 1 -- 3.

If $A\ls\kappa < B$, then the motion takes place in the region $\B $ and the equations admit the first integrals
\begin{equation}\label{neq3}
  H_w=p_w^2+hW(w)=h\kappa,\qquad H_\varphi=p_\varphi^2-h\Phi(\varphi)=-h\kappa,
\end{equation}
where $W(w)=U(u(w))$, $\Phi(\varphi)=V(v(\varphi))$. The qualitative picture of the functions $W$ and $\Phi$ is shown in Fig.~\ref{fig_02}. In Fig.~\ref{fig_03}, we show the phase portraits of  one-dimensional systems corresponding to the Hamilton functions $H_w$ and $H_\varphi$. Each manifold $J_{h, \kappa}$ is the product of level lines of the functions $H_w$ and $H_\varphi$ defined by \eqref{neq3}. Thus, $J_{h,A}$ is two non-intersecting circles (they correspond to the cross section of the ellipsoid by the plane $x=0$ with two different directions of motion). If $A<\kappa<B$, then $J_{h,\kappa}$ consists of two two-dimensional tori each of which concentrically envelopes one of the circles out of $J_{h,A}$.

\begin{figure}[!ht]
\centering
\includegraphics[width=0.6\paperwidth,keepaspectratio]{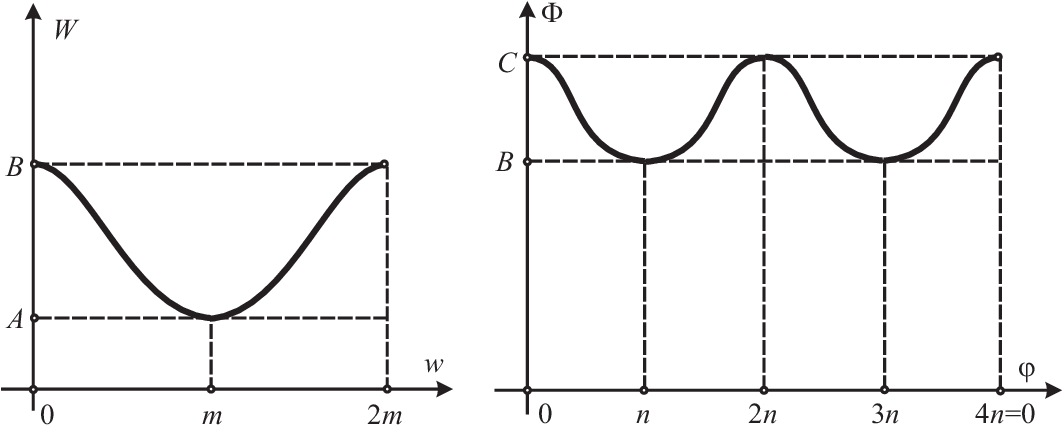}
\caption{The functions $W$ and $\Phi$.}\label{fig_02}
\end{figure}

\begin{figure}[!ht]
\centering
\includegraphics[width=0.6\paperwidth,keepaspectratio]{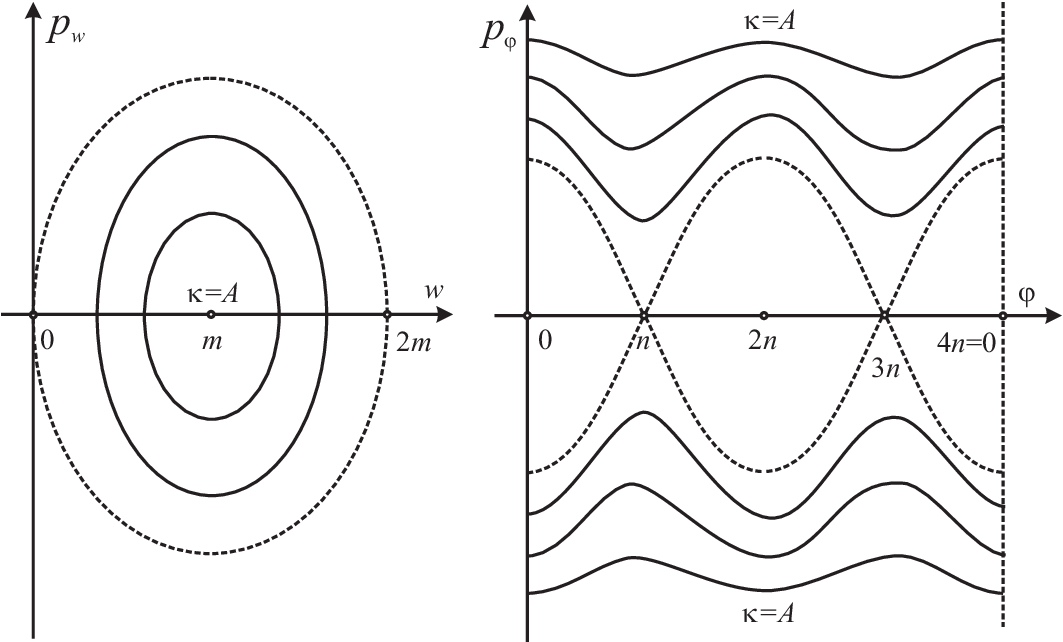}
\caption{The portraits of one-dimensional systems.}\label{fig_03}
\end{figure}

In the case $B<\kappa\ls C$ the motion takes place in the region $\C $. In $T^*\C  \cap J_h$ the integrals are defined
$$
H_s=p_s^2-hS(s)=-h\kappa, \qquad H_\theta=p_\theta^2+h\Theta(\theta)=h\kappa,
$$
where $S(s)=V(v(s))$, $\Theta(\theta)=U(u(\theta))$; the system splits into two one-dimensional ones. The manifold $J_{h,C}$ consists of two non-intersecting circles and $J_{h,\kappa}$ for $B<\kappa<C$ consists of two two-dimensional tori each of which concentrically envelopes one of the circles out of $J_{h,C}$.

In Fig.~\ref{fig_04}, where the diametrically opposite points of the ball boundary are identified, we show the sets corresponding to the manifolds $J_{h,A}$ and $J_{h,C}$ under the homeomorphism $\beta:P\to J_h$. The union of the circles {\it 1} and {\it 2} is the set $\beta^{-1}(J_{h,C})$. The set $\beta^{-1}(J_{h,A})$ consists of the circles {\it 3} and~{\it 4}.

\begin{figure}[!ht]
\centering
\includegraphics[width=0.3\paperwidth,keepaspectratio]{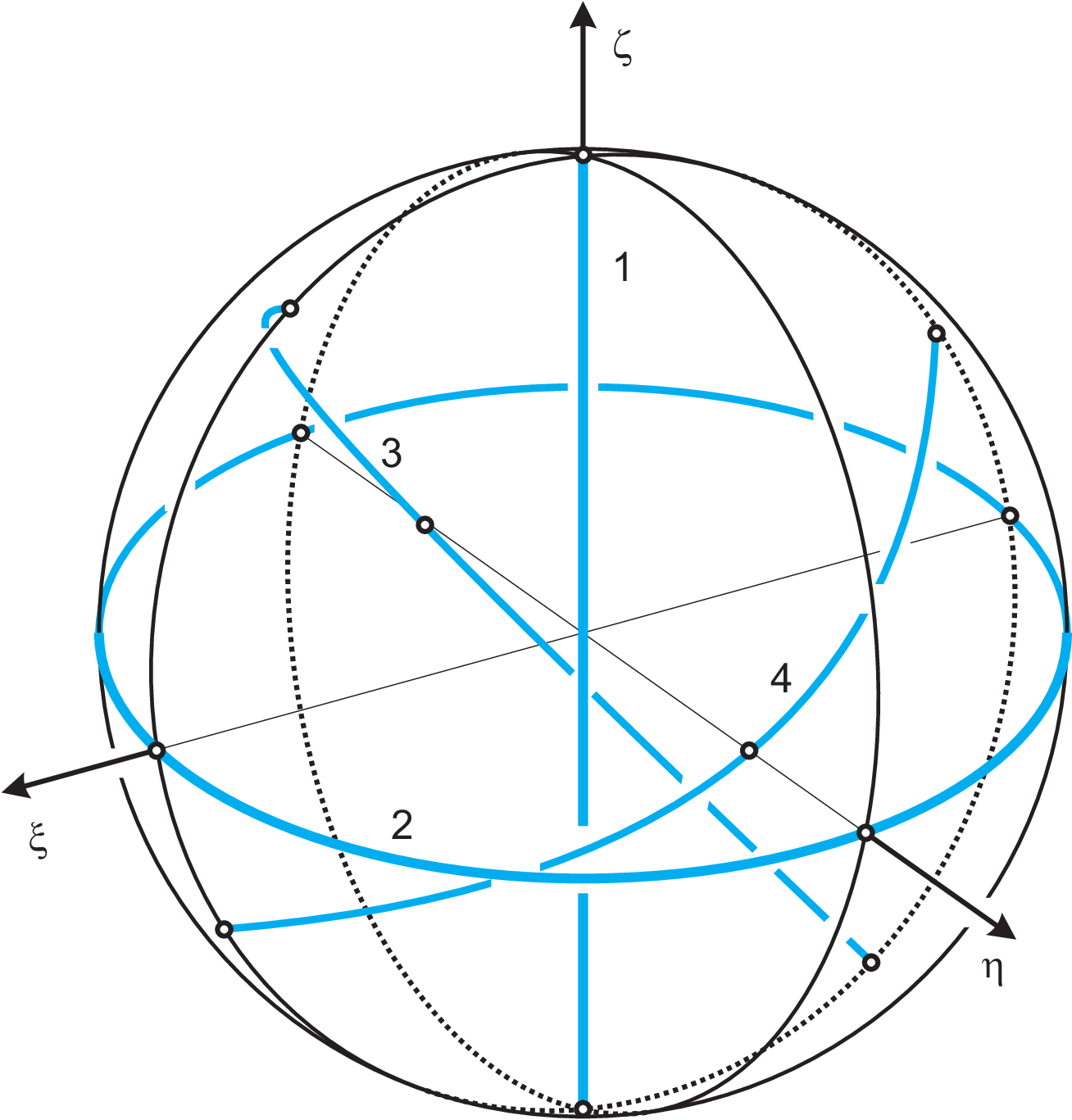}
\caption{The integral circles.}\label{fig_04}
\end{figure}

\begin{figure}[!ht]
\centering
\includegraphics[width=0.3\paperwidth,keepaspectratio]{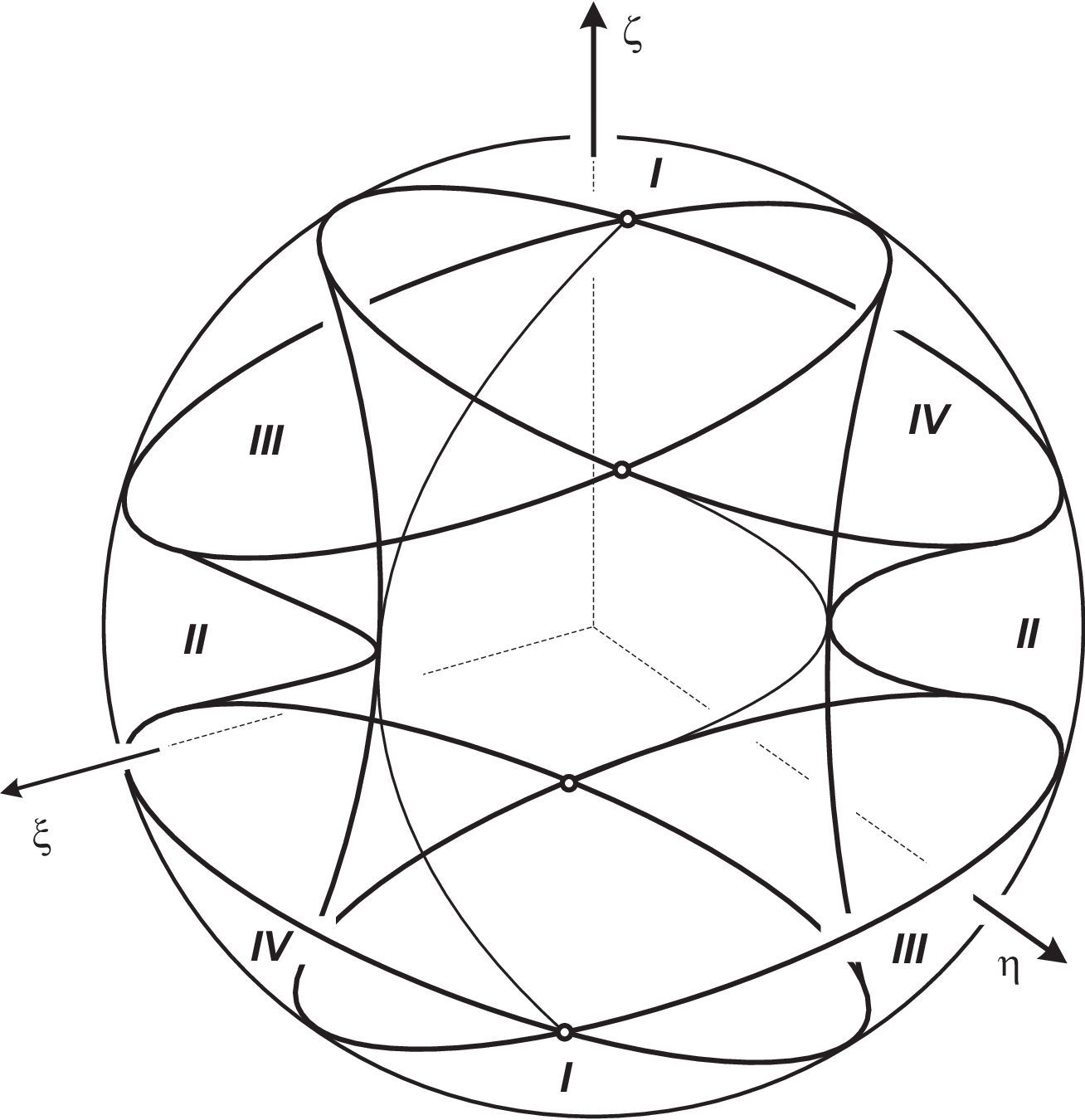}
\caption{The tori regions.}\label{fig_05}
\end{figure}

Now let us consider the case $\kappa = B$. We denote by $K_1$, $K_2$, $K_3$, and $K_4$ the umbilical points $(u=0, v=n)$ on the ellipsoid surface lying respectively in the regions $\{x>0, z>0\}$, $\{x<0, z>0\}$, $\{x<0, z<0\}$, and $\{x>0, z<0\}$.

\begin{propos}\label{propos2}
The cross section of the ellipsoid by the plane $y=0$ is a closed geodesic of the metric $d\Sigma$. All geodesics starting from an umbilical point at $t=0$ meet simultaneously at the opposite umbilical point.
\end{propos}

\begin{proof} Let us use the coordinates $(w, \varphi)$. Introducing the ``reduced time'' $\tau$ by the formula $d\tau=[\Phi(\varphi)-W(w)]^{-1}dt$ and using equations \eqref{neq3} with $\kappa=B$, we get the equations of geodesics in the form
\begin{equation}\label{neq4}
  \ds{\frac{dw}{d\tau}=\pm\sqrt{h(B-W(w))}, \qquad \frac{d\varphi}{d\tau}=\pm\sqrt{h(\Phi(\varphi)-B)}.}
\end{equation}
Denote
$$
\ds{F(w,w_0)=\int\limits_{w_0}^w\frac{dw}{\sqrt{h(B-W(w)})}}, \quad  \ds{G(\varphi,\varphi_0)=\int\limits_{\varphi_0}^\varphi\frac{d\varphi}{\sqrt{h(\Phi(\varphi)-B)}}}.
$$
Let $w=f(\tau, w_0)$ and $\varphi=g(\tau, w_0)$ be the inverse for the dependencies $\tau=F(w,w_0)$ and $\tau=G(\varphi,\varphi_0)$ respectively. Equations \eqref{neq4} admit the solutions
$$
\begin{array}{ll}
\bigl(w\equiv 0,\varphi=g(\pm\tau,\varphi_0)\bigr), & \bigl(w\equiv 2m,\varphi=g(\pm\tau,\varphi_0)\bigr), \\ \bigl(w=f(\pm\tau,w_0),\varphi\equiv n\bigr), & \bigl(w=f(\pm\tau,w_0),\varphi\equiv 3n\bigr).
\end{array}
$$
This proves the first statement.

Consider an arbitrary trajectory of equations \eqref{neq4} starting at a point $\{w_0,\varphi_0\}$ not belonging to the cross section $y=0$. Let, for definition, this point lie in the first octant, i.e., $m<w_0<2m$, $0<\varphi_0<n$. The initial velocity may have four directions according to the choice of the signs in \eqref{neq4}. Suppose, for example, that $\left. dw/d\tau \right|_{\tau=0}>0$, $\left. d\varphi/d\tau \right|_{\tau=0}>0$. Then (see Fig.~\ref{fig_03}) as $\tau \to +\infty$, the coordinates $w$ and $\varphi$ monotonously increase and $w \to 2m$, $\varphi\to n({\rm mod}\, 4n)$. As $\tau \to -\infty$ we have monotonous decreasing $w \to 0$ and $\varphi\to -n({\rm mod}\, 4n)$.

Therefore the chosen trajectory of \eqref{neq4} asymptotically approaches $K_1$ as $\tau \to +\infty$ and $K_3$ as $\tau \to -\infty$. Another possible cases of the inial directions are considered analogously.

So, since the geodesics starting at an umbilical point can correspond only to the value $\kappa = B$, each such geodesic meets the cross section $y=0$ for the first time at the opposite umbilical point.

Let $\gamma_1(t)$ and $\gamma_2(t)$ be two geodesics such that $\gamma_1(0) = \gamma_2(0) = K_3$. Suppose that some time value $t=t_0>0$ corresponds to the value $\tau=0$ of the ``reduced time''. Let $\gamma_1(t_0)=(w_1, \varphi_1)$, $\gamma_2(t_0)=(w_2, \varphi_2)$. Then the dependency of $\gamma_1$ on the ``reduced time'' is $w=f(\tau, w_1)$, $\varphi=g(\tau, \varphi_1)$, and the equations of $\gamma_2$ are $w=f(\tau, w_2)$, $\varphi=g(\tau, \varphi_2)$. Denote by $t_1$ and $t_2$ the minimal positive values of $t$ for which $\gamma_1(t_1) = \gamma_2(t_2) = K_1$. Then
\begin{equation}\label{neq5}
  t_1=\int\limits_{-\infty}^{+\infty}[\Phi(g(\tau, \varphi_1))-W(f(\tau, \varphi_1))]d\tau,
\end{equation}
\begin{equation}\label{neq6}
  t_2=\int\limits_{-\infty}^{+\infty}[\Phi(g(\tau, \varphi_2))-W(f(\tau, \varphi_2))]d\tau.
\end{equation}
The integrals in \eqref{neq5} and \eqref{neq6} converge since the metric $d\Sigma$ does not have singularities.

Let us show that $t_1=t_2$. For this purpose we use the obvious relations
\begin{equation}\label{neq7}
  f(\tau,w_1)=f(\tau-F(w_1,w_2),w_2),\qquad g(\tau,w_1)=g(\tau-G(w_1,w_2),w_2)
\end{equation}
and the following almost obvious statement. Suppose that for a function $\psi(\tau)$ $(-\infty<\tau<+\infty)$ there exists such a point $\tau_0$ that $\chi(\tau)=\psi(\tau+\tau_0)$ is an even function. If the integral
$$
\int\limits_{-\infty}^{+\infty}[\psi(\tau)-\psi(\tau+k)]d\tau,
$$
with some constant $k$ converges, then it equals zero. Using \eqref{neq7}, we transform \eqref{neq5} as follows
\begin{equation*}
  t_1=\int\limits_{-\infty}^{+\infty}[\Phi(g(\tau,\varphi_2))-W(f(\tau+G(\varphi_2,\varphi_1))-F(w_2,w_1),w_2)]d\tau.
\end{equation*}
Then we subtract the equality \eqref{neq6}:
$$
t_1-t_2=\int\limits_{-\infty}^{+\infty}[W(f(\tau,w_2))-W(f(\tau+k,w_2))]d\tau.
$$
Here
$$
k=G(\varphi_2, \varphi_1)-F(w_2,w_1)
$$
does not depend on $\tau$.

It is easy to check that $W(f(\tau,w_2))$ as a function of $\tau$ satisfies the condition of the just formulated statement. For this, it is sufficient to choose $\tau_0$ in such a way that $f(\tau_0,w_2)=m$. Consequently, $t_1=t_2$. The proposition is proved.
\end{proof}

Let us now describe the type of the set $J_{h,B}$. The curves $O_i=J_h \cap T_K^*\Ellips^2$ $(i=1,2,3,4)$ are topological circles. According to Proposition~\ref{propos2}, all trajectories starting at $O_1$ simultaneously cross $O_3$ and simultaneously return to $O_1$. Therefore this family of trajectories fills a closed flow tube, i.e., they fill a two-dimensional torus $T_1$ in $J_h$. In the same way the family of geodesics crossing $K_2$ and $K_4$ fills a two-dimensional torus $T_2$ in $J_h$. The tori $T_1$ and $T_2$ intersect by two circles corresponding to the cross section of the ellipsoid by the plane $y=0$ with two deffirent directions of motion.

In Fig.~\ref{fig_05}, we show how the set $\beta^{-1}(J_{h,B})$ is embedded in $P$ (the diametrically opposite points of the ball boundary are identified). The regions $I-IV$ are filled with the one-parameter families of the integral tori enveloping concentrically the circles {\it 1~--~4} respectively (see Fig.~\ref{fig_04}).

\end{document}